\documentclass[11pt]{article}
\usepackage{fullpage,amsthm,amssymb,amsmath}
\usepackage{graphicx,algorithmic,algorithm,amsfonts, verbatim}
\usepackage{enumerate}
\usepackage{color}
\usepackage{blindtext,caption}
\usepackage{tikz}

\newtheorem{theorem}{Theorem}[section]
\newtheorem{lemma}[theorem]{Lemma}

\newtheorem{proposition}[theorem]{Proposition}

\newcommand\abs[1]{\lvert #1\rvert}

\newcommand{\mm}{\operatorname{mm}}
\newcommand{\bw}{\operatorname{bw}}
\newcommand{\tw}{\operatorname{tw}}
\newcommand{\mmw}{\operatorname{mmw}}

\begin{document}

%\title{Hardness of computing maximum matching width}
\title{Computing the maximum matching width is NP-hard}
\author{Kwangjun Ahn and Jisu Jeong\thanks{
Department of Mathematical Sciences, Korea Advanced Institute of Science and Technology, Daejeon, Republic of Korea.
\texttt{kjahnkorea@kaist.ac.kr, jjisu@kaist.ac.kr}}}
%\date\today
\date{\vspace{-5ex}}

\maketitle

\begin{abstract}
The maximum matching width is a graph width parameter that is defined on a branch-decomposition over the vertex set of a graph. In this short paper, we prove that the problem of computing the maximum matching width is NP-hard.
\end{abstract}

\section{Introduction}

Tree-width and branch-width are two prominent graph width parameters extensively studied in both structural
graph theory and theoretical computer science due to wide applications. For instance, the seminar work by Courcelle~\cite{courcelle1990monadic} proved that many NP-hard problems in graphs can be solved in polynomial time when the graph has bounded tree-width or branch-width, proving the usefulness the width parameters.

 Recently, a new graph width parameter compatible with tree-width or branch-width was introduced, called the \emph{maximum matching width} (\emph{mm-width} for short)~\cite{vatshelle2012new}.  
 It is shown in~\cite{vatshelle2012new} that 
 if $\mmw(G)$, $\tw(G)$, and $\bw(G)$ are the mm-width, tree-width, and branch-width of a graph $G$, respectively,
 then, for every graph $G$, 
 \[
 \mmw(G) \leq \max(\bw(G), 1) \leq  \tw(G) + 1 \leq 3\mmw(G).
 \] 
% for every graph $G$, where $\mmw(G)$, $\tw(G)$, and $\bw(G)$ are the mm-width, tree-width, and branch-width of $G$, respectively.
 Because of the similarity between the parameters one might anticipate that the new parameter can be replaced by the existing width parameters, questioning the necessity of studying the new parameter.
On the other hand, it has been observed that the new parameter can lead to better results than others when employed to design algorithms for popular problems~\cite{kloks2005computing}.
In particular, it is recently shown that the algorithm for the Minimum Dominating Set problem based on mm-width can enjoy faster time-complexity than one based on tree-width~\cite{JST2015}.

In this short paper, we investigate the hardness of computing the mm-width. 
 While it is widely known that the problem of computing the tree-width or branch-width in general graphs are both NP-hard~\cite{arnborg1987complexity,seymour1994call}, it has not been explored whether or not computing mm-width is NP-hard. We prove that computing the mm-width is also NP-hard by leveraging recent results in~\cite{JST2015}.   

\subsection{Notations}

All graphs are simple, undirected, and finite.
For a graph $G$, let $V(G)$ and $E(G)$ denote the vertex set and edge set of the graph, respectively.
For a subset of vertices  $S\subseteq V(G)$, $N_G(S)$ is the subset of vertices that are adjacent to at least one vertex in $S$.
A tree is called \emph{nontrivial} if it has at least one edge. 
We say that a tree is \emph{ternary} if all vertices have degree $1$ or $3$.
A function $f:2^X\rightarrow \mathbb{Z}$ is \emph{symmetric} if $f(A)=f(X\setminus A)$ for all $A\subseteq X$, and
a function $f$ is \emph{submodular} if $f(A)+f(B)\ge f(A\cup B)+f(A\cap B)$ for all $A,B\subseteq X$.

\section{Preliminaries}

\subsection{Maximum matching width}

We provide a formal definition of the maximum matching width. 
A \emph{branch-decomposition} over a finite set $X$ is a pair $(T,\mathcal{L})$ of a ternary tree $T$ together with a bijection $\mathcal{L}$ from the leaves of $T$ to $X$.
%Hence, here and below, leaves will be identified with $X$.
When we remove an edge $ab$ of $T$, the tree is divided into two connected components, 
inducing  the partition $X_1\cup X_2$ of $X$ due to bijection $\mathcal{L}$. 
Here, we call $X_1\cup X_2$ the \emph{partition induced by $ab$}.

For an edge $e$ of $T$, and a function $f$, which is symmetric and submodular, the \emph{$f$-value} of $e$ is $f(A)(=f(B))$ where $(A,B)$ is the partition induced by $e$.
The \emph{$f$-width} of a branch-decomposition $(T,\mathcal{L})$ is the maximum of $f$-values over $E(T)$.
The \emph{$f$-width} of $X$ is the minimum value of the $f$-width over all possible branch-decompositions over $X$.

Let $\mm_G:2^{V(G)} \rightarrow \mathbb{Z}$ be a function, where $\mm_G(A)$ is defined as the size of a maximum matching in $G$ between $A$ and its complement. 
Note that the function $\mm_G$ is symmetric and submodular~\cite{DBLP:journals/corr/SaetherT14}.
The \emph{maximum matching width} of $G$, denoted by $\mmw(G)$, is the $\mm_G$-width of~$V(G)$.

Jeong, Telle, and S{\ae}ther~\cite{JST2015} gave a new characterization of graphs whose mm-width are at most $k$, formally stated as the following theorem. 

\begin{theorem}[Jeong, Telle, and S{\ae}ther~{\cite[Theorem 3.8]{JST2015}}] \label{jeongetal}
	%(Theorem 3.8. in~\cite{JST2015}) 
	A nontrivial graph $G$ has $mmw(G)\leq k$ if and only if there exist a ternary tree $T$ and nontrivial subtrees $T_u$ of $T$ for all vertices $u\in V(G)$ such that
	\begin{enumerate}[(1)]
		\item if $uv\in E(G)$ then the subtrees $T_u$ and $T_v$ have at least one node of $T$ in common, and
		\item for every edge of $T$ there are at most $k$ subtrees using this edge.
	\end{enumerate}
\end{theorem}

A \emph{tree-representation} of a graph $G$ is a pair $(T, \{T_x\}_{x\in V(G)} )$ where $T$ is a
ternary tree and a collection $\{T_x\}_{x\in V(G)}$ of nontrivial subtrees of $T$ 
satisfying the property $(1)$. 
Theorem~\ref{jeongetal} states that a graph G has a tree-representation in which every edge of $T$ is contained in at most $k$ subtrees if and only if $\mmw(G) \leq k$.

\subsection{Helly property of subtrees}

A set system $\mathcal{F}$ is said to satisfy the \emph{Helly property} if the following holds for every subcollection $\mathcal{G}\subseteq \mathcal{F}$: 
\[
\text{if }A\cap B \neq \emptyset\text{ for all }A,B\in \mathcal{G},\text{ then }\bigcap_{A\in \mathcal{G}}A \neq \emptyset.
\]
It is well known that a collection of the node sets of subtrees of a tree satisfies the \emph{Helly property}:
\begin{proposition} \label{lem:helly}
	Let $C$ be a clique of a graph $G$ with at least two vertices. For every tree representation $(T, \{T_x\}_{x\in V(G)} )$, $\bigcap_{x\in C} V(T_x) \not = \emptyset$.
\end{proposition}

\section{Computing maximum matching width is NP-hard}

	A graph $G$ is called a \emph{split graph} if $V(G)$ can be partitioned into an independent set $I$ of $G$ and a clique $C$ of $G$, in which case %we also denote by $G=(I\cup C, E)$.
	we write $C=C(G)$ and $I=I(G)$.
Inspired by~\cite{kloks2005computing}, we prove, in particular, that computing the mm-width of split graphs is NP-hard.

The following lemma characterizes the conditions satisfied by a certain type of split graphs with mm-width equal to $k$:
\begin{lemma} \label{lem:key}
%Suppose $G=(I\cup C,E)$ is a split graph with $\abs{C}=3k$, $k\ge 1$.
Let $G$ be a split graph with $\abs{C(G)}=3k$, $k\ge 1$.
Then
$\mmw(G)=k$ if and only if
%$C$ partitions into three subsets $C_1$, $C_2$, and $C_3$ with $\abs{C_1}=\abs{C_2}=\abs{C_3}= k$ 
%such that for every vertex $w\in I$ there is an index $\ell\in\{1,2,3\}$ with $N_G(w)\subseteq C_\ell$.
$C(G)$ can be partitioned into three subsets $C_1$, $C_2$, and $C_3$ with $\abs{C_1}=\abs{C_2}=\abs{C_3}= k$ 
such that, for each vertex $w\in I(G)$, $N_G(w)$ is contained exactly one of $C_1$, $C_2$, and $C_3$.
\end{lemma}

\begin{proof}
($\Rightarrow$) Let $C=C(G)$. Assume that $(T, \{T_x\}_{x\in V(G)} )$ is a tree-representation of $G$ in which every edge of $T$ is contained in at most $k$ subtrees, whose existence is ensured by Lemma~\ref{jeongetal}. By Proposition~\ref{lem:helly}, there exists a vertex $v_0 \in \bigcap_{x\in C} V(T_x)$ in $T$. Then $v_0$ cannot be a leaf; otherwise, the unique edge incident with $v_0$ would be contained in $3k$ subtrees $T_x$ for all $x\in C$. Hence, the degree of $v_0$ is $3$, and let $e_1,e_2,e_3$ be the three incident edges in $T$. 
Let $s_i$ be the number of subtrees containing $e_i$ for every $i=1,2,3$. 
Then, $s_i\le k$ for every $i=1,2,3$.
For each $x\in C$, $T_x$ contains at least one edge among $e_1,e_2,e_3$
because it contains $v_0$. Hence, $s_1+s_2+s_3 \geq 3k$ and 
%On the other hand, $s_i\leq k$ for each $i=1,2,3$, so 
it follows that $s_i= k$  for each $i=1,2,3$.
This also implies that each tree $T_x$ with $x\in C$ contains exactly one edge among $e_1,e_2,e_3$.
Therefore, by defining $C_i:=\{x\in C : e_i \in E(T_x) \}$, one can partition $C$ into $C_1 \cup C_2 \cup C_3$
with $\abs{C_1}=\abs{C_2}=\abs{C_3}= k$.

Note that %$T\setminus v_0$ 
the graph obtained from $T$ by deleting $v_0$
consists of three disconnected components. 
Denote by $T^{(i)}$ the component that is incident with $e_i$ for each $i=1,2,3$. 
For each $w\in I(G)$, $T_w$ contains none of $e_1,e_2,e_3$ because the edges are fully occupied by subtrees corresponding to vertices in $C$. 
Thus, $T_w$ should be entirely contained in $T^{(j)}$ for some $j \in \{1,2,3\}$. 
This implies that, for each $x\in C$, if $V(T_w)\cap V(T_{x})\not = \emptyset$, 
then $x\in C_j$. In other words, $N_G(w) \subseteq C_j$.

{($\Leftarrow$)} Assume that $C(G)$ admits a partition $C_1\cup C_2\cup C_3$ satisfying the property. Let $C_j=\{c^{(j)}_1,c^{(j)}_2,\ldots,c^{(j)}_k\}$ for every $j=1,2,3$.
Moreover, assume that  $I(G)$ is partitioned into $I_1\cup I_2\cup I_3$ so that for each $j=1,2,3$, $N_G(I_j)\subseteq C_j $. Let $I_j=\{i^{(j)}_1,\ldots, i^{(j)}_{\ell_j}\}$ for every $j=1,2,3$.
	For every tree-representation $(T, \{T_x\}_{x\in V(G)} )$, Proposition~\ref{lem:helly} ensures that there exists a vertex $v_0 \in \bigcap_{x\in V(C)} V(T_x)$ in $T$. As $T$ is ternary, $v_0$ is incident with at most three edges, and from the pigeonhole principle, at least one edge should be contained in at least $k$ subtrees. Thus, ${\mmw(G)\geq k}$.
	
Now, we show that $\mmw(G)\leq k$. It is enough to construct a tree-representation $(T, \{T_x\}_{x\in V(G)} )$ of $G$ where every edge in $T$ is contained in at most $k$ subtrees by Lemma~\ref{jeongetal}. 
First, introduce $|V(G)| + (|V(G)|-3) +1$ many vertices $\{\beta_{x}\}_{x\in V(G)}$, $\{\alpha_{x}\}_{x\in \left(V(G)\setminus \{c_k^{(1)},c_k^{(2)},c_k^{(3)}\}\right)}$ and $\alpha_0$.
We construct a ternary tree $T$ as follows:
\begin{enumerate}[(a)]
	\item Build three paths %(line graphs) 
	$\alpha_{i^{(j)}_1} \alpha_{i^{(j)}_2} \cdots \alpha_{i^{(j)}_{\ell_j}}\alpha_{c^{(j)}_1}\alpha_{c^{(j)}_2}\cdots \alpha_{c^{(j)}_{k-1}}$ for every $j=1,2,3$.
	\item Join the three paths by adding three edges $\{\alpha_{i^{(j)}_1},\alpha_0 \}$ for every $j=1,2,3$.
	\item For $x\in \left(V(G)\setminus \{c_k^{(1)},c_k^{(2)},c_k^{(3)}\}\right)$, add an edge between $\beta_x$ and $\alpha_x$. In addition, attach $\beta_{c_k^{(j)}}$ to $\alpha_{c_{k-1}^{(j)}}$ for every $j=1,2,3$.
\end{enumerate}
In this way, we obtain a ternary tree whose leaves are $\beta_{x}$'s. 
%%%%%%%%%%%%%%%%%
\begin{figure}
	\centering
	\includegraphics[width=\textwidth]{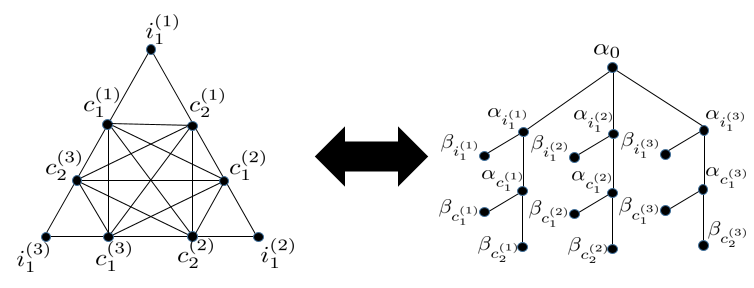}
	\caption[An example of the construction of a tree-representation]{An example of the construction of a tree-representation when $k=2$, $\ell_j=1$ for every $j=1,2,3$.	 For each $w\in I(G)$, $T_w$ is the subtree consisting of a single edge $\{\alpha_w, \beta_w\}$. For each $c \in C(G)$, $T_c$ is the unique path from $\beta_{c}$ to $\alpha_0$.
		\label{fig1}}
\end{figure}
%%%%%%%%%%%%%%%%%
See Figure~\ref{fig1} for an example. 
Now we define the collection $\{T_x\}_{x\in I(G)\cup C(G)}$ of subtrees as follows: 

\begin{itemize}
	\item For each $w\in I(G)$, $T_w$ is the subtree consisting of a single edge $\{\alpha_w, \beta_w\}$.
	\item For each $c \in C(G)$, $T_c$ is the unique path from $\beta_{c}$ to $\alpha_0$. 
\end{itemize}
Then, it is straightforward to check that this construction yields a desired tree-representation.  
\end{proof}

We now introduce a decision problem, which we call \emph{PARTITION-3}: Given a multi-set $S$ (a set in which multiple elements are allowed) of positive integers, the task is to decide whether $S$ can be partitioned into three multi-subsets $S_1\cup S_2 \cup S_3$ such that $\sum_{s\in S_1}s = \sum_{s\in S_2}s = \sum_{s\in S_3}s$.

For instance, when $S=\{ 3,1,1,2,1 \}$, the answer is NO as the sum of elements is not even divisible by $3$; when $S=\{ 4, 4,7 \}$, the answer is NO as every subset containing the element $7$ will exceed a third of the total sum; when $S=\{ 1,1,1, 3, 2,1\}$, the answer is YES as $S_1=\{1,1,1\}$, $S_2=\{3\}$, $S_3=\{2,1\}$ gives a desired partition.

\begin{lemma}
	PARTITION-3 is NP-hard.
\end{lemma}
\begin{proof}
	We construct a polynomial reduction from \emph{PARTITION} (problem [SP12] in~\cite{garey2002computers}): The instance of \emph{PARTITION} is the same as \emph{PARTITION-3}, and the task is to decide whether a multiset can be partitioned into two multi-subsets of equal sums rather than three. The reduction is constructed as follows: for a given instance $S=\{s_1,\ldots ,s_m\}$ of \emph{PARTITION}, construct an instance of \emph{PARTITION} as $S' = \{s_1,s_2,\ldots,s_m, \frac{1}{2}\sum_{i=1}^m s_i\}$. The correctness of this reduction is  straightforward.
	Because \emph{PARTITION} is known to be NP-hard~\cite{garey2002computers}, \emph{PARTITION-3} is NP-hard.
\end{proof}

Now, we finish the proof.
\begin{theorem}
	Computing the maximum matching width for graphs is NP-hard. In particular, computing the maximum matching width for split graphs is NP-hard.
\end{theorem}
\begin{proof}
	The reduction is from \emph{PARTITION-3}. For a given instance $S=\{s_1,\ldots ,s_m\}$, we construct a split graph $G$ as follows: 
	\begin{enumerate}
		\item Consider a complete graph on $\sum_{i=1}^m s_i$ vertices $V$. 
		Partition $V$ into $m$ many subsets $V_1 \cup V_2 \cup \cdots \cup V_m$ so that $|V_i|=s_i$ for each $i=1,2,\ldots ,m$.
		\item Introduce $m$ more vertices $I=\{i_1,\ldots,i_m\}$, and connect $i_j$ to all vertices in $V_j$ for each $j=1,2,\ldots,m$.
		%Partition $C$ into $m$ many subsets $C_1 \cup C_2 \cup \cdots \cup C_m$ so that $|C_i|=s_i$ for each $i=1,2,\ldots ,m$.\item Introduce $m$ more vertices $I=\{i_1,\ldots,i_m\}$, and connect $i_j$ to all vertices in $C_j$ for each $j=1,2,\ldots,m$.
	\end{enumerate}
By Lemma~\ref{lem:key}, such constructed split graph has mm-width equal to $\frac{1}{3}\sum_{i=1}^m s_i$ if and only if $V$ can be partitioned into three equal-sized partition $C_1\cup C_2 \cup C_3$ such that for each $w\in I$, $N_G(w)$ is completely contained in one of three partitions. 
By letting $I_\ell := \{ s_w : N_G(w) \subseteq C_{\ell}  \}$ for every $\ell=1,2,3$, one can see that $\sum_{w\in I_\ell } s_w =\frac{1}{3}\sum_{i=1}^m s_i $.
\end{proof}

\section*{Acknowledgements}
We thank Jan Arne Telle, Sigve Hortemo S{\ae}ther, and Yixin Cao for fruitful advice.

\bibliographystyle{abbrv}
\bibliography{mmwnph}

\begin{thebibliography}{1}

\bibitem{arnborg1987complexity}
S.~Arnborg, D.~G. Corneil, and A.~Proskurowski.
\newblock Complexity of finding embeddings in a {$k$}-tree.
\newblock {\em SIAM J. Algebraic Discrete Methods}, 8(2):277--284, 1987.

\bibitem{courcelle1990monadic}
B.~Courcelle.
\newblock The monadic second-order logic of graphs. {I}. {R}ecognizable sets of
  finite graphs.
\newblock {\em Inform. and Comput.}, 85(1):12--75, 1990.

\bibitem{garey2002computers}
M.~R. Garey and D.~S. Johnson.
\newblock {\em Computers and intractability}.
\newblock W. H. Freeman and Co., San Francisco, Calif., 1979.
\newblock A guide to the theory of NP-completeness, A Series of Books in the
  Mathematical Sciences.

\bibitem{JST2015}
J.~Jeong, S.~H. S{\ae}ther, and J.~A. Telle.
\newblock {Maximum Matching Width: New Characterizations and a Fast Algorithm
  for Dominating Set}.
\newblock In T.~Husfeldt and I.~Kanj, editors, {\em 10th International
  Symposium on Parameterized and Exact Computation (IPEC 2015)}, volume~43 of
  {\em Leibniz International Proceedings in Informatics (LIPIcs)}, pages
  212--223, Dagstuhl, Germany, 2015. Schloss Dagstuhl--Leibniz-Zentrum fuer
  Informatik.

\bibitem{kloks2005computing}
T.~Kloks, J.~Kratochv\'\i~l, and H.~M\"uller.
\newblock Computing the branchwidth of interval graphs.
\newblock {\em Discrete Appl. Math.}, 145(2):266--275, 2005.

\bibitem{DBLP:journals/corr/SaetherT14}
S.~H. S{\ae}ther and J.~A. Telle.
\newblock Between treewidth and clique-width.
\newblock {\em Algorithmica}, 75(1):218--253, 2016.

\bibitem{seymour1994call}
P.~D. Seymour and R.~Thomas.
\newblock Call routing and the ratcatcher.
\newblock {\em Combinatorica}, 14(2):217--241, 1994.

\bibitem{vatshelle2012new}
M.~Vatshelle.
\newblock {\em New Width Parameters of Graphs}.
\newblock PhD thesis, University of Bergen, 2012.

\end{thebibliography}

\end{document}